\documentclass[final,3p]{elsarticle}
\usepackage{amsmath}
\usepackage{amsthm}
\usepackage{amsfonts}
\usepackage{amssymb}
\usepackage{epstopdf}
\usepackage{mathrsfs}
\usepackage{graphicx}
\usepackage{cite}
\DeclareMathOperator*{\argmin}{arg\,min}

\newtheorem{definition}{Definition}
\newtheorem{lemma}{Lemma}
\newtheorem{theorem}{Theorem}

\newcommand{\norm}[1]{\left|\left|#1\right|\right|_2}

\begin{document}

\begin{frontmatter}



\title{On the Number of Iterations for Convergence of CoSaMP and Subspace Pursuit Algorithms}


\author[a]{Siddhartha Satpathi}
\author[b]{Mrityunjoy
Chakraborty}

\address{Department of Electronics and Electrical Communication Engineering, Indian Institute of Technology, Kharagpur, INDIA}
\address[a]{sidd.piku@gmail.com}
\address[b]{mrityun@ece.iitkgp.ernet.in}

\begin{abstract}
In compressive sensing, one important parameter that characterizes
the various greedy recovery algorithms is the iteration bound
which provides the maximum number of iterations by which the
algorithm is guaranteed to converge. In this letter, we present a
new iteration bound for CoSaMP by certain mathematical
manipulations including formulation of appropriate sufficient
conditions that ensure passage of a chosen support through the two
selection stages of CoSaMP, ``Augment" and ``Update".
Subsequently, we extend the treatment to the subspace pursuit (SP)
algorithm. The proposed  iteration bounds for both CoSaMP and SP
algorithms are seen to be improvements over their existing
counterparts, revealing that both CoSaMP and SP algorithms
converge in fewer iterations than suggested by results available
in literature.
\end{abstract}


\end{frontmatter}

\section{Introduction}

Reconstruction of signals in compressed sensing (CS)
\citep{Donoho} involves obtaining the sparsest solution to an
underdetermined set of equations given as ${\bf y}={\bf A}{\bf
x},$ where ${\bf A}$ is an $m\times l$ ($m<<l$) complex valued,
sensing matrix and $y$ is an $m\times 1$ complex valued
observation vector. It is assumed that the sparsest solution to
the above system is $K$-sparse, i.e., not more than $K$ (for some
minimum $K,K>0$) elements of ${\bf x}$ are non-zero and also that
the sparsest solution is unique, which can be guaranteed if every
$2K$ columns of $\bf A$ are linearly independent \citep{Tropp_2K}.
Greedy approaches like orthogonal matching pursuit (OMP)
\citep{OMP_Tropp}, compressive sampling matching pursuit (CoSaMP)
\citep{Needell}, subspace pursuit (SP) \citep{Dai}, hard
thresholding pursuit (HTP) \citep{HTP} and others recover the
$K$-sparse signal by iteratively constructing the support set of
the sparse signal (i.e., index of non-zero elements in the sparse
vector) by some greedy principles. These greedy pursuits are well
known for their low complexity.

Convergence of these iterative procedures in a finite number of
steps requires the matrix $\bf A$ to satisfy the so-called
``Restricted Isometry Property (RIP)" \citep{Elad} of appropriate
order as given below.
\begin{definition}
 A matrix ${\bf A}\in \mathbb{C}^{m\times l}\;(m<l)$ is said to satisfy the RIP of order $K$
if there exists a ``Restricted Isometry Constant (RIC)"
$\delta_K\in (0,\;1)$ so that
\begin{equation}
(1-\delta_K)\norm{{\bf x}}^2\leqslant \norm{{\bf A x}}^2\leqslant(1+\delta_K)\norm{{\bf x}}^2
\end{equation}
for all $K$-sparse ${\bf x}$. The constant $\delta_K$ is taken as
the smallest number from $(0,\;1)$ for which the RIP is satisfied.
\end{definition}


Convergence of the CS greedy recovery algorithms is usually
established by imposing certain upper bounds on the RIC $\delta_K$
as a sufficient condition. In the case of CoSaMP, such a bound is
given by $\delta_{4K} < 0.5$ \citep{Song}, which is a refined
(i.e., more relaxed) version of two earlier bounds, namely,
$\delta_{4K} < 0.17157$ \citep{Needell} and $ \delta_{4K} <
0.38427$ \citep{FoucartCoSa}. Similarly, for SP, the bound
proposed originally is $ \delta_{3K} < 0.205$ \citep{Dai}, which
was improved afterwards to $ \delta_{3K} < 0.325$ \citep{SP2} and
$ \delta_{3K} < 0.4859$ \citep{Song}.

Apart from the convergence bound on the RIC, there is another
important parameter that characterizes a greedy algorithm, namely,
the iteration bound, which provides the maximum (finite) number of
iterations by which the algorithm is guaranteed to converge. For
CoSaMP, a signal independent iteration bound of $6(K+1)$ was
presented in \citep{Needell}, assuming $\delta <0.1$. In this
letter, we present a new iteration bound for CoSaMP which refines
the above result and is given as \textit{a function of  $\delta$
over the entire range for which convergence of CoSaMP is currently
guaranteed  (i.e., $0<\delta<0.5$)}.
For this, we first develop a sufficient condition for capturing
the support of the $p+q$ largest (in magnitude) elements of $\bf
x$ within certain number of iterations ($0<p+q\le K$), given that
the support for the $p$ largest elements of $\bf x$ has already
been captured. The derivation takes appropriate steps so that the
above sufficient condition is obtained in a form structurally
similar to the one proposed earlier for the HTP algorithm
\citep{FoucartHTP}. This permits computation of the iteration
bound via a procedure suggested in \citep{FoucartHTP}.
Subsequently, we extend our approach to the SP algorithm and
compute the corresponding iteration bound which is seen to be
tighter than existing results on this \citep{Dai} for more
practical ranges of the RIC and is thus an improvement, as it
establishes that the SP algorithm in more practical cases
converges in fewer iterations than suggested in \citep{Dai}.
\section{Notations and a brief review of the CoSaMP \& the SP algorithms}

We denote by $Z$ the index set $\{1,\,2,\cdots,\,l\}$. Then, given
$B\subseteq Z$ and ${\bf z}\in \mathbb{C}^l$, the vector ${\bf
z}_B\in\mathbb{C}^l$ is defined as follows : $[{\bf z}_B]_i=[{\bf
z}]_i$ for $i\in B$ and $[{\bf z}_B]_i=0$ otherwise. Similarly,
given the matrix ${\bf A}\in \mathbb{C}^{m\times l}$, the matrix
${\bf A}_{B}\in\mathbb{C}^{m\times l}$ is defined such that for
$i\in B$, $[{\bf A}_{B}]_i=[{\bf A}]_i$ (where $[.]_i$ denotes the
$i$-th column of the matrix $[.]$) and $[{\bf A}_{B}]_i={\bf 0}$
otherwise. The notation $supp(.)$ denotes the support of the
vector $(.)$, i.e., $supp({\bf z})=\{i\in Z\;|\;z_i\ne 0\}$. By
$S$ and $S^k$, we denote respectively the true support set of
${\bf x}$ and the estimated support set after $k$ iterations.
Elements of the vector $|{\bf x}|$ sorted in descending order form
the vector $\widetilde{{\bf x}}$ and the $r$-th element of
$\widetilde{{\bf x}},~r=1,\,2,\cdots,l$ is denoted by
$\widetilde{x}_r$, i.e., $\widetilde{x}_r=[\widetilde{{\bf
x}}]_r$. The index of the $i^{th}$ largest (in magnitude) element
of ${\bf x}$, $i=1,2,\cdots,|S|$ is denoted by $\pi(i)$, implying
$|{\bf x}_{\pi(i)}|=\widetilde{x}_i$. Lastly, for a matrix $\bf
A$, ${\bf A}^h$ denotes its Hermitian transposition.

For convenience of presentation, we adopt the convention of using
the notation :
$\stackrel{(.)}{=}$
to indicate that the equality ``$=$" follows from
Equation (.)
(same for inequalities). Also, unless stated otherwise, the more
generalized form of CS will be considered in this paper where $\bf
x$ is $K$-sparse but $\bf y$ is contaminated with a noise vector
$\bf e\in \mathbb{C}^{m}$, i.e., ${\bf y}={\bf A}{\bf x}+{\bf e}$.
\begin{table}[h]
\caption {Compressive Sampling Matching Pursuit algorithm}
\begin{tabular}{p{15cm}}
\hline \textbf{Input}: measurement ${\bf y}\in\mathbb{C}^m$,
sensing matrix ${\bf A}\in\mathbb{C}^{m\times l}$,
 sparsity $K$, stopping error $\epsilon$, initial estimate ${\bf x^0}$\\
\hline\\
\textbf{For} ( $n=1$ ; $|| {\bf y}-{\bf A}{\bf x^{n-1}}||_2>\epsilon$ ; $n=n+1$ )\\
\hspace{7mm}\textit{\textbf{Identification}}: $\Gamma^{n}=supp(H_{2K}({\bf A}^h({\bf y}-{\bf A}{\bf x^{n-1}})))$\\
\hspace{7mm}\textit{\textbf{Augment}}: $U^{n}=S^{n-1} \cup \Gamma^{n}:S^{n-1}=supp({\bf x^{n-1}})$\\
\hspace{7mm}\textit{\textbf{Estimate}}: ${\bf
u^{n}}=\underset{{\bf z}:supp({\bf
z})=U^{n}}\argmin \;||{\bf y}-{\bf A}{\bf z}||_2$\\
\hspace{7mm}\textit{\textbf{Update}}: ${\bf x^{n}}=H_K({\bf u^{n}})$\\
\textbf{Output}: ${\bf \hat {x}}={\bf x^{n-1}}$
\\
\hline \label{CoSaMP}
\end{tabular}
\end{table}

The CoSaMP and the SP algorithms are given in Table \ref{CoSaMP}
and \ref{SP} respectively. Both algorithms iteratively estimate
$\bf x$, with ${\bf x^{n}}$ denoting the estimate at the $n$-th
iteration. At the ``Identification" stage, in both algorithms, the
residue vector $({\bf y}-{\bf A}{\bf x^{n-1}})$ is first
correlated with the columns of $\bf A$. The support $\Gamma^{n}$
of the top $r$ elements ($r=2K$ in case of CoSaMP and $r=K$ in
case of SP) in terms of magnitude of correlations is then
identified, using a hard thresholding operator $H_{r}(.)$ that
retains the top $r$ (in magnitude) elements of the vector $(.)$
and sets other elements to zero. The vector $\bf y$ is then
projected orthogonally on the column space of ${\bf A}_{S^{n-1}
\cup \Gamma^{n}}$ where $S^{n-1}=supp({\bf x^{n-1}})$, generating
the projection coefficient vector ${\bf u^{n}}$. In CoSaMP, the
new estimate ${\bf x^{n}}$ is taken as $H_K({\bf u^{n}})$, whereas
in SP, $\bf y$ is further projected on the column space of ${\bf
A}_{S^{n}}$, where $S^n=supp(H_K({\bf u^{n}}))$, and ${\bf x^{n}}$
is taken as the corresponding projection coefficient vector.

\begin{table}
\caption {Subspace Pursuit Algorithm}
\begin{tabular}{p{15cm}}
\hline \textbf{Input}: measurement ${\bf y}\in\mathbb{C}^m$,
sensing matrix ${\bf A}\in\mathbb{C}^{m\times l}$,
 sparsity $K$, stopping error $\epsilon$, initial estimate ${\bf x^0}$\\
\hline\\
\textbf{For} ( $n=1$ ; $|| {\bf y}-{\bf A}{\bf x^{n-1}}||_2>\epsilon$ ; $n=n+1$ )\\
\hspace{7mm}\textit{\textbf{Identification}}: $\Gamma^{n}=supp(H_{K}({\bf A}^h({\bf y}-{\bf A}{\bf x^{n-1}})))$\\
\hspace{7mm}\textit{\textbf{Augment}}: $U^{n}=S^{n-1} \cup \Gamma^{n}:S^{n-1}=supp({\bf x^{n-1}})$\\
\hspace{7mm}\textit{\textbf{Estimate}}: ${\bf
u^{n}}=\underset{{\bf z}:supp({\bf
z})=U^{n}}\argmin \;||{\bf y}-{\bf A}{\bf z}||_2$\\
\hspace{7mm}\textit{\textbf{Update}}: $S^n=supp(H_K({\bf u^{n}}))$\\
\hspace{17mm}${\bf x^{n}}=\underset{{\bf z}:supp({\bf z})=S^{n}}\argmin \;||{\bf y}-{\bf A}{\bf z}||_2$\\
\textbf{Output}: ${\bf \hat {x}}={\bf x^{n-1}}$
\\
\hline \label{SP}
\end{tabular}
\end{table}
\section{Proposed Iteration Bound Analysis for CoSaMP and SP Algorithms}
\subsection{Iteration Bound for CoSaMP}
The proposed iteration bound computation for CoSaMP depends on the
dynamics of decay of $||{\bf x}_{\overline{U^n}}||_2$ over $n$
(under appropriate conditions on the RIC), which is presented in
Lemma 1 below and is obtained by introducing suitable
modifications in the corresponding analysis in \citep{Song}, which
considers decay of $||{\bf x}-{\bf x}^n||_2$ (rather than $||{\bf
x}_{\overline{U^n}}||_2$) over $n$.
\begin{lemma}
In CoSaMP algorithm, the metric $||{\bf x}_{\overline{U^n}}||_2$
decays over $n$ with the rate
$\rho_{4K}=\sqrt{\frac{2\delta^2_{4K}(1+2\delta^2_{4K})}{1-\delta^2_{4K}}}$
as per the following :
\begin{equation}
||{\bf
x}_{\overline{U^n}}||_2<\rho_{4K}||{\bf
x}_{\overline{U^{n-1}}}||_2+(1-\rho_{4K})\tau\norm{\textbf{e}}
\label{metric1}
\end{equation}
where the constant $\tau$ is defined using $\rho_{4K}$ as,
$(1-\rho_{4K})\tau
=\frac{\delta_{4K}\sqrt{6(1+\delta_{3K})}}{1-\delta_{4K}}+\sqrt{2(1+\delta_{4K})}$.
\end{lemma}
\begin{proof}
Given in Appendix A.
\end{proof}
Using the above Lemma, we next derive a sufficient condition for
capturing the support of, say, the $p+q$ largest elements of $\bf
x$ in $k$ iterations ($0<p+q\le K$), assuming that the support of
the $p$ largest elements of $\bf x$ has already been captured
(where by ``largest", we mean largest in magnitude). In
particular, we strive to obtain the above sufficient condition in
a form that is structurally identical to the one developed for the
HTP algorithm in Lemma 3 of \citep{FoucartHTP}, so that the
procedure to compute the iteration bound as presented in
\citep{FoucartHTP} can be applied. This is, however, not easy
considering that algorithmically, CoSaMP is substantially
different from the HTP algorithm and is in particular
characterized by certain steps like ``Augment" (i.e., expansion of
the support set to size $3K$, as given in Table I) and ``Update"
(i.e., pruning the support set to size $K$ from $3K$) not present
in HTP.
 In Theorem 1 below, we show how the above can be achieved by deploying
 suitable mathematical manipulations, and in particular, by formulating appropriate sufficient
conditions that ensure that the support of the $p+q$ largest
elements of $\bf x$ gets selected in both the ``Augment" step and
the subsequent ``Update" step.

\begin{theorem}
 Assume that at the $n$-th iteration in
the CoSaMP algorithm, $S^n:supp({\bf x^n})$ contains the support
of the $p~(p<K)$ largest (in magnitude) entries of ${\bf x}$.
Then, a sufficient condition for capturing the support of the
$p+q$ largest (in magnitude) entries of ${\bf x}$ in $k$
additional iterations for some integer $q$, $K-p\ge q \ge 1$ is
given by
\begin{equation}
\widetilde{x}_{p+q}>\rho_{4K}^k||\widetilde{{\bf
x}}_{\{p+1,p+2,\cdots ,K\}}||_2+\gamma||e||_2,~(\rho_{4K}<1)
\end{equation}
where $\gamma$ is a function of $\delta_{3K}$ and $\delta_{4K}$.
\end{theorem}

\begin{proof}
We need to ensure that the support of the largest (in magnitude)
$p+q$ elements of ${\bf x}$, i.e., $\{\pi(1),\cdots,\pi(p+q)\}$
gets selected in the $(n+k)^{th}$ iteration. This means,
$\pi(j),\;j\in \{ 1,\cdots ,p+q\}$ should first belong to
$U^{n+k}$, and then it also should go through the update step in
CoSaMP. Now, for $\pi(j),~j\in \{1,\cdots,p+q\}$ to belong to
$U^{n+k}$, it is sufficient to have
\begin{equation}
\widetilde{x}_{p+q}>||({\bf x})_{\overline{U^{n+k}}}||_2,\label{eq1}
\end{equation}
as this ensures that the top $p+q$ elements of $\bf x$ can not
belong to $({\bf x})_{\overline{U^{n+k}}}$ and thus, their support
is captured in $U^{n+k}$. In order that the above support passes
through the update step in CoSaMP under the satisfaction of
\eqref{eq1}, it is sufficient to have,
\begin{align}
&\min_{j\in \{\pi(1),\cdots ,\pi(p+q)\}}{|{\bf u^{n+k}}_j|}>\max_{i \in U^{n+k} \setminus S} {|{\bf u^{n+k}}_i|}.\label{CoSa4}
\end{align}
In the following, we first develop a sufficient condition (viz.
(9)) which, under the satisfaction of (4), guarantees satisfaction
of (5). Condition (3) is then obtained by deriving a sufficient
condition for \textit{simultaneous} satisfaction of (4) and (9).

Note that one can write $|{\bf u^{n+k}}_j|=|{\bf x}_j-({\bf
x}_j-{\bf u^{n+k}}_j)|\ge |{\bf x}_j|-|{\bf u^{n+k}}_j-{\bf
x}_j|.$ Using this and some basic properties of inequalities, the
LHS of (\ref{CoSa4}) can be written as,
\begin{align}
&\nonumber\min_{j\in \{\pi(1),\cdots ,\pi(p+q)\}}{|{\bf u^{n+k}}_j|}\\
&\ge
\min_{j\in \{\pi(1),\cdots ,\pi(p+q)\}}{|({\bf x})_j|
-|({\bf u^{n+k}-x})_j|}\nonumber\\
& \ge \widetilde{x}_{p+q}-\max_{j\in \{\pi(1),\cdots ,\pi(p+q)\}}|({\bf
u^{n+k}-x})_j|,
\end{align}
while the RHS of (\ref{CoSa4}) can be written as $\max_{i \in
U^{n+k} \setminus S} {|{\bf u^{n+k}}_i|}=\max_{i \in U^{n+k}
\setminus S} {|({\bf u^{n+k}}-{\bf x})_i|}$, since ${\bf x}_i=0$
for $i\in U^{n+k}
\setminus S$. Combining, in order to have (\ref{CoSa4}) satisfied,
it is then sufficient to have,
\begin{align}
&\widetilde{x}_{p+q}>\max_{j\in \{\pi(1),\cdots ,\pi(p+q)\}}|({\bf u^{n+k}}-{\bf x})_j|+\max_{i \in U^{n+k} \setminus S}{|({\bf u^{n+k}}-{\bf x})_i|}.
\end{align}
Now, under the satisfaction of \eqref{eq1}, we have $\pi(j)\in
U^{n+k},~j\in \{1,\cdots,p+q\}$. Again, $\pi(j)\in S,~j\in
\{1,\cdots,p+q\}$. Together, these mean that $\pi(j)\in
U^{n+k}\cap S,~j\in \{1,\cdots,p+q\}$. One can then write,
$$\max_{j\in \{\pi(1)\cdots ,\pi(p+q)\}}|({\bf u^{n+k}}-{\bf x})_j|\le ||({\bf u^{n+k}}-{\bf x})_{U^{n+k}\cap S}||.$$
Also,
$$\max_{i \in U^{n+k} \setminus S}{|({\bf u^{n+k}}-{\bf x})_i|}\le ||({\bf u^{n+k}}-{\bf x})_{U^{n+k} \setminus S}||.$$
Using these and the fact that for two real numbers $a,\;b$,
$a+b\le \sqrt{2}\sqrt{a^2+b^2}$, the RHS of (7) can be written as,
\begin{align}
&\nonumber\max_{j\in \{\pi(1),..,\pi(p+q)\}}|({\bf u^{n+k}}-{\bf x})_j|+\max_{i \in U^{n+k} \setminus S}{|({\bf u^{n+k}}-{\bf x})_i|}\\
&\nonumber<\sqrt{2}||({\bf u^{n+k}-x})_{U^{n+k}}||_2
\\
&\stackrel{(\ref{estimation})}{\le}\dfrac{\sqrt{2}\delta_{4K}}{\sqrt{1-\delta_{4K}^2}}||({\bf x})_{\overline{U^{n+k}}}||_2+\sqrt{2}\tau_1\norm{\textbf{e}}.
\end{align}
From (7) and (8), it then follows that it is sufficient to have,
\begin{equation}
\widetilde{x}_{p+q}>\dfrac{\sqrt{2}\delta_{4K}}{\sqrt{1-\delta_{4K}^2}}||({\bf x})_{\overline{U^{n+k}}}||_2+\sqrt{2}\tau_1\norm{\textbf{e}},\label{eq2}
\end{equation}
in order to satisfy \eqref{CoSa4} under the condition that
\eqref{eq1} holds. Now, to satisfy both (\ref{eq1}) and
(\ref{eq2}) simultaneously, a sufficient condition will be
$\widetilde{x}_{p+q}>\max ($RHS of \eqref{eq1}, RHS of
\eqref{eq2}$)$. Recalling that convergence of CoSaMP requires
$\delta_{4K}<0.5$ \citep{Song} which implies
$\dfrac{\sqrt{2}\delta_{4K}}{\sqrt{1-\delta_{4K}^2}}<1$, it will
then be enough to have $\widetilde{x}_{p+q}>||({\bf
x})_{\overline{U^{n+k}}}||_2+\sqrt{2}\tau_1\norm{\textbf{e}}$ for
simultaneous satisfaction of (\ref{eq1}) and (\ref{eq2}), where,
\begin{align}
&||({\bf x})_{\overline{U^{n+k}}}||_2+\sqrt{2}\tau_1\norm{\textbf{e}}\nonumber
\\
&\stackrel{\eqref{metric1}}{<}\rho_{4K}^{k}||({\bf x})_{\overline{U^{n}}}||_2+(\tau +\sqrt{2}\tau_1)\norm{\textbf{e}}\nonumber
\\
&\le\rho_{4K}^{k}||({\bf x})_{\overline{S^{n}}}||_2+\gamma \norm{\textbf{e}}\nonumber
\\
&\le\rho_{4K}^{k}||\widetilde{{\bf x}}_{\{p+1,\cdots ,K\}}||_2+\gamma \norm{\textbf{e}},
\end{align}
where $\gamma=\tau +\sqrt{2}\tau_1$ and the last step follows from the assumption that $S^n$ has captured
the $p$ largest (in magnitude) elements of $\bf x$.

Hence, in order that the $p+q$ largest (in magnitude) elements of $\bf x$ get selected in $k$ additional
iterations, it is sufficient to have,
\begin{equation}
\widetilde{x}_{p+q}>\rho_{4K}^{k}||\widetilde{{\bf x}}_{\{p+1,\cdots ,K\}}||_2+\gamma \norm{\textbf{e}}.
\end{equation}
Hence proved.
\end{proof}

Note that by substituting $p=0$ and $q=K$ in Theorem 1 and taking
$\norm{\textbf{e}}$ to be zero, one can obtain the minimum number
of iterations required to guarantee perfect recovery in the
noiseless case (i.e., iteration bound), which is given by
$k_{min}=\lceil\frac{{\log(\norm{\bf
x}}/\widetilde{x}_K)}{\log(1/\rho_{4K})}\rceil.$
The above, however, provides a $k_{min}$ that is dependent on the
signal structure, i.e., $\norm{\bf x}$ and $\widetilde{x}_K$. A
signal independent iteration bound can, however, be computed by
noticing the similarity between (2) and the corresponding
sufficient condition for the HTP algorithm derived in Lemma 3 of
\citep{FoucartHTP} (with the only difference being in the
expressions for $\rho_{4K}$ and $\gamma$). To calculate $k_{min}$,
one then simply has to apply the arguments used in the proof of
Theorem 5 of \citep{FoucartHTP} to the above context, which will
require successive application of Theorem 1 on certain partitions
of the index set $S$. The resulting iteration bound is given in
Theorem 2 below.
\begin{theorem}
With measurements ${\bf y}={\bf Ax}$, the CoSaMP algorithm converges to $\bf x$ in $\lceil cK \rceil$ number of iterations where
$
c=\ln{(4/\rho^2_{4K})}/\ln{(1/\rho^2_{4K})}.
$
\end{theorem}
\begin{proof}
The proof follows directly by applying the arguments used in the
proof of Theorem 5 of \citep{FoucartHTP} to (2) and is thus
omitted.
\end{proof}

Note that unlike \citep{Needell} where an iteration bound for
CoSaMP was calculated assuming $0<\delta_{4K} <0.1$, the proposed
bound is defined for $0<\delta_{4K}<0.5$, i.e., over the entire
range for which CoSaMP is guaranteed to converge. In order to have
some quantitative idea, we plot the proposed iteration bound
(after normalizing by $K$) against $\delta_{4K}$ in Fig. 1.
Clearly, in comparison to \citep{Needell} which obtained the
iteration bound as $6(K+1)$ (for $0<\delta_{4K}<0.1$), the
proposed bound is about four to six times less, which is a
significant improvement.
\begin{figure}[h]
\begin{center}
\includegraphics[width=120mm,height=80mm]{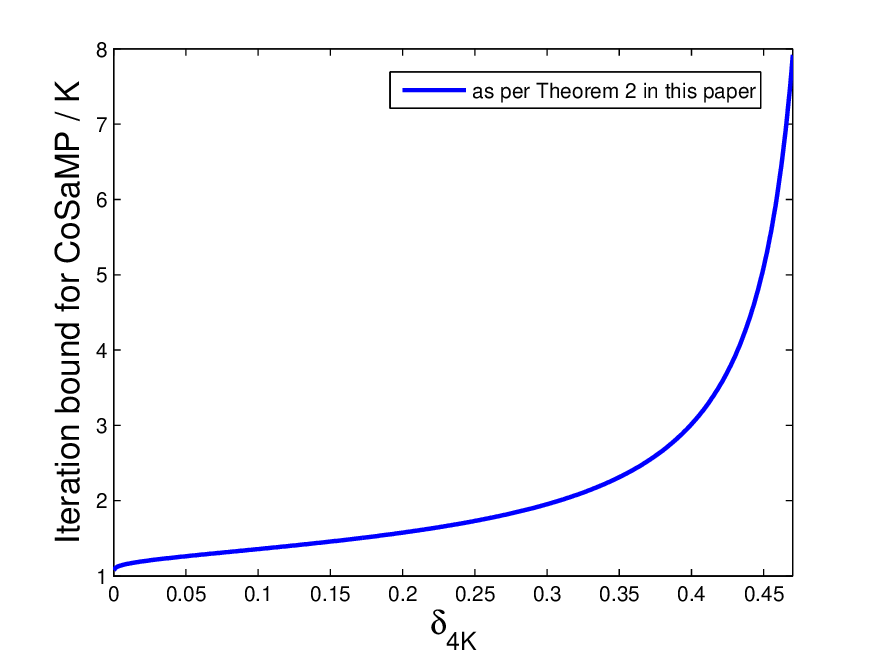}
\end{center}
\caption{Dependence of the proposed iteration bound for CoSaMP (after normalizing by $K$)
on the RIC $\delta_{4K}$} 
\end{figure}
\subsection{Extension to the Subspace Pursuit Algorithm}
Like Lemma 1 for CoSaMP above, there exists a similar decay
relation for the SP algorithm as well \citep{Song}. However, due
to the presence of an additional orthogonal projection step in the
SP algorithm (viz., the second operation in the ``Update" step),
the decay relation is obtained here directly in terms of
$\norm{{\bf x}_{\overline{S^n}}}$. This makes it possible to
formulate a sufficient condition to capture the $p+q$ largest (in
magnitude) elements of $\bf x$ within a certain number of
iterations in a much simpler way than in CoSaMP and thus the
derivation becomes lot simpler.

\begin{theorem}
With measurements ${\bf y}={\bf Ax}$, the SP algorithm converges
to ${\bf x}$ in $\lceil cK \rceil$ number of iterations, where $
c=\ln{(4/\rho^2_{3K})}/\ln{(1/\rho^2_{3K})}. $
\label{mainThSP}
\end{theorem}
\begin{proof}
In the SP algorithm, a decay relation  analogous to Lemma 1 for
CoSaMP is given by \citep{Song}
\begin{equation}
\norm{{\bf x}_{\overline{S^n}}}<\rho_{3K}\norm{{\bf
x}_{\overline{S^{n-1}}}},\label{decaySP}
\end{equation}
where
$\rho_{3K}=\frac{\sqrt{2\delta_{3K}^2(1+\delta_{3K}^2)}}{1-\delta_{3K}^2}$.
As before, we now develop conditions to ensure that the support of
the largest (in magnitude) $p+q$ elements of ${\bf x}$, i.e.,
$\{\pi(1),\cdots,\pi(p+q)\}$ get selected in $S^{n+k}$, assuming
that the support $\{\pi(1),\cdots,\pi(p)\}$ has already been
selected in $S^{n}$. A sufficient condition to ensure that the
support $\{\pi(1),\cdots,\pi(p+q)\}$ is captured in $S^{n+k}$ will
be given by,
%
\begin{equation}
\widetilde{x}_{p+q}>\norm{{\bf x}_{\overline{S^{n+k}}}}. \label{eqSP1}
\end{equation}
From \eqref{decaySP}, proceeding recursively backwards, we can
write $\norm{{\bf x}_{\overline{S^{n+k}}}}<\rho_{3K}^k\norm{{\bf
x}_{\overline{S^{n}}}}$. Again, from the assumption that the
support $\{\pi(1),\cdots,\pi(p)\}$ has already been selected in
$S^{n}$, we have, $\norm{({\bf x})_{\overline{S^{n}}}}
\le\norm{(\widetilde{{\bf x}})_{\{p+1,\cdots ,K\}}}$. From this
and \eqref{eqSP1}, a sufficient condition for ensuring that
$\{\pi(1),\cdots,\pi(p+q)\}\subseteq S^{n+k}$ given that
$\{\pi(1),\cdots,\pi(p)\}\subseteq S^{n}$ is given by,
\begin{equation}
\widetilde{x}_{p+q}>\rho_{3K}^k\norm{(\widetilde{{\bf x}})_{\{p+1,\cdots ,K\}}}.
\label{eqSP2}
\end{equation}
Since \eqref{eqSP2} has the same form as that of Lemma 3 of
\citep{FoucartHTP}, one can compute the iteration bound by
directly applying the procedure given in Theorem 5 of
\citep{FoucartHTP}. The resulting iteration bound is given by
$\lceil cK \rceil$, where
$c=\frac{\ln{(4/\rho^2_{3K})}}{\ln{(1/\rho^2_{3K})}}$.
\end{proof}

For higher and thus more practical values of $\delta_{3K}$, the
iteration bound proposed in Theorem 3 is an improvement over the
existing result $\lceil \frac{1.5K}{\ln{(1/\rho_{3K})}} \rceil$ as
given in Theorem 6 of \citep{Dai}. To show this, we plot both the
iteration bounds (after normalizing by $K$)
 against $\delta_{3K}$ in Fig. 2 over the range
$0<\delta_{3K}<0.4859$, i.e., the range for which the SP algorithm
is guaranteed to converge, after adopting $\rho_{3K}$ from
\citep{Song}. It is seen from Fig. 2 that while for
$0<\delta_{3K}<0.28$, the proposed iteration bound is slightly
higher than that of \citep{Dai}, for $\delta_{3K}>0.28$ (which is
also a more practical range for $\delta_{3K}$ for the SP
algorithm), the former is significantly lesser than the latter and
the difference grows with $\delta_{3K}$. This shows that for more
practical ranges of $\delta_{3K}$, the SP algorithm actually
converges in significantly fewer iterations than suggested in
\citep{Dai}.

\begin{figure}[htb]

\begin{minipage}[b]{1.0\linewidth}
  \centering
  \centerline{\includegraphics[width=12cm]{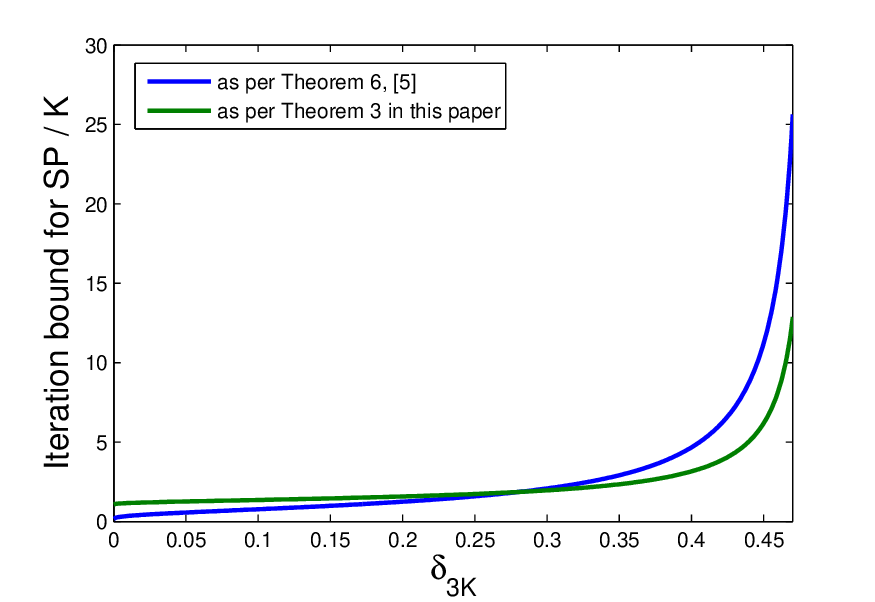}}
\end{minipage}
 \caption{Comparison of the proposed iteration
bound of the SP algorithm vis-a-vis the existing iteration bound
as given in \citep{Dai}.}
\end{figure}
%

%
\appendix

\section{Proof of Lemma 1}
As a consequence of the identification step in CoSaMP, as proved in Lemma 7, \citep{Song}, we can write,
\begin{equation}
\norm{({\bf x})_{\overline{U^n}}}\le\sqrt{2}\delta_{4K}\norm{{\bf x}-{\bf x^{n-1}}}+\sqrt{2(1+\delta_{4K})}\norm{\textbf{e}},
\label{A3}
\end{equation}
Now, we need to upper bound $\norm{{\bf x}-{\bf x^{n-1}}}$ in
terms of $\norm{({\bf x})_{\overline{U^{n-1}}}}$. For notational
convenience, we obtain the upper bound of $\norm{{\bf x}-{\bf
x^{n}}}$ in terms of $\norm{({\bf x})_{\overline{U^{n}}}}$ and
then replace $n$ with $n-1$ later. For this, the estimation step
of CoSaMP is analyzed next. For any ${\bf z}\in \mathbb{C}^{l
\times 1}$ with $supp({\bf z})\subseteq U^n$, the estimation step
in CoSaMP ensures that $\langle {\bf y} - {\bf A u^n},{\bf
Az}\rangle=0$. Substituting $\bf y$ by ${\bf A}{\bf x}+{\bf e}$,
this leads to $\langle {\bf u^n} - {\bf x} ,{\bf A}^h{\bf
Az}\rangle=\langle{\bf e} ,{\bf Az}\rangle$. Defining $V^n=U^n\cup
S$, this can also be written as,
\begin{align}
\langle {\bf u^n} - {\bf x} ,{\bf A}^h_{V^n}{\bf A}_{V^n}{\bf
z}\rangle=\langle{\bf e} ,{\bf A}{\bf z}\rangle .
\label{estimation1}
\end{align}
Taking $z=({\bf u^n}-{\bf x})_{U^n}$, one then obtains,
\begin{align}
&\nonumber||({\bf u^n}-{\bf x})_{U^n}||_2^2=\langle {\bf u^n}-
{\bf x} , ({\bf u^n}-{\bf x})_{U^n} \rangle,
\\
&\nonumber\stackrel{(\ref{estimation1})}{=}\langle {\bf u}^n-{\bf x},
(I-{\bf A}^h_{V^n}{\bf A}_{V^n})({\bf u^n}-{\bf x})_{U^n}
\rangle+\langle {\bf e} , {\bf A}({\bf u^n}-{\bf x})_{U^n}
\rangle,
\\
&\le\delta_{4K}||{\bf u^n}-{\bf x}||_2||({\bf u^n}-{\bf x})_{U^n}||_2+\sqrt{1+\delta_{3K}}||{\bf e}||_2||({\bf u^n}-{\bf x})_{U^n}||_2,\label{A0}
\end{align}
where \eqref{A0} follows from consequences of RIP as presented in Lemma 1,2 \citep{Song}.
With $A=||({\bf u^n}-{\bf x})_{U^n}||_2$, $B=||({\bf u^n}-{\bf x})_{\overline{U^n}}||_2$ and $||{\bf u^n}-{\bf x}||_2^2=A^2+B^2$, \eqref{A0} can be reduced to $A\le\delta_{4K}\sqrt{A^2+B^2}+\sqrt{1+\delta_{3K}}||{\bf e}||_2$. After solving the above quadratic equation in $A$ with appropriate inequalities, we get,
\begin{equation}
||({\bf u^n}-{\bf x})_{U^n}||_2 \le \dfrac{ \delta_{4K} }{\sqrt{1-\delta_{4K}^2}}||({\bf u^n}-{\bf x})_{\overline{U^n}}||_2+\tau_1\norm{\textbf{e}},
\label{estimation}
\end{equation}
where $\tau_1=\dfrac{\sqrt{1+\delta_{3K}}}{1-\delta_{4K}}$. Coming
to the update step, as $S^n$ is the best $K$ term approximation to
$U^n$, we can say
\begin{equation}
\norm{{\bf u^n}_{A\setminus B }}\le\norm{{\bf u^n}_{B\setminus A}}=\norm{({\bf u^n}-{\bf x})_{B\setminus A}},\label{update1}
\end{equation}
where $A=U^n\setminus S^n$ and $B=U^n\setminus S$. We also have,
\begin{align}
&\norm{{\bf u^n}_{A\setminus B}}=\norm{({\bf u^n}-{\bf x})_{A\setminus B}+({\bf x})_{A}}
\\
&\ge \norm{({\bf x})_{U^n\setminus S^n}}-\norm{({\bf u^n}-{\bf
x})_{A\setminus B}}.\label{update2}
\end{align}
From (\ref{update1}) and (\ref{update2}), we get,
\begin{equation}
\norm{({\bf x})_{U^n\setminus S^n}}\le
\sqrt{2}\norm{({\bf u^n}-{\bf x})_{A\cup B}}\le\sqrt{2}\norm{({\bf u^n}-{\bf x})_{U^n}}.\label{update}
\end{equation}
Finally, we can upper bound $\norm{{\bf x}-{\bf x^{n}}}$ as,
\begin{align}
&\nonumber\norm{{\bf x}-{\bf x^{n}}}^2=\norm{({\bf x}-{\bf x^{n}})_{S^n}}^2+\norm{({\bf x}-{\bf x^{n}})_{\overline{S^n}}}^2
\\
&\nonumber=\norm{({\bf x}-{\bf x^{n}})_{S^n}}^2+\norm{({\bf x})_{\overline{U^n}}}^2+\norm{({\bf x})_{U^n\setminus S^n}}^2
\\
&\nonumber\stackrel{(\ref{update})}{\le}\norm{({\bf x}-{\bf x^{n}})_{S^n}}^2+\norm{({\bf x})_{\overline{U^n}}}^2+2\norm{({\bf u^n}-{\bf x})_{U^n}}^2
\\
&\nonumber\le 3\norm{({\bf x}-{\bf u^{n}})_{U^n}}^2+\norm{({\bf
x})_{\overline{U^n}}}^2
\\
&\nonumber\stackrel{(\ref{estimation})}{\le}\big{(}\dfrac{\sqrt{3}\delta_{4K}}{\sqrt{1-\delta_{4K}^2}}\norm{({\bf x})_{\overline{U^n}}}+\sqrt{3}\tau_1\norm{\textbf{e}}\big{)}^2+\norm{({\bf x})_{\overline{U^n}}}^2
\\
&\le(\sqrt{\dfrac{1+2\delta_{4K}^2}{1-\delta_{4K}^2}}\norm{({\bf x})_{\overline{U^n}}}+\sqrt{3}\tau_1\norm{\textbf{e}})^2.\label{A1}
\end{align}
Using \eqref{A1} with \eqref{A3}, we then arrive at the result.

\textbf{Comment :} Lemma 1 and its proof as given above has some
important differences with its counterpart presented in
\citep{Song} (i.e., Theorem 2 of \citep{Song}). In \citep{Song}, a
similar decay relation was presented in terms of $\norm{{\bf
x}-{\bf x^n}}$. For this, $\norm{({\bf x}-{\bf x^n})_{S^n}}$ and
$\norm{({\bf x}-{\bf x^n})_{\overline{S^n}}}$ were separately
upper bounded by $\norm{{\bf x}-{\bf x^{n-1}}}$, and then the
results were combined to obtain an upper bound of $$\norm{{\bf
x}-{\bf x}^{n}}=\sqrt{\norm{({\bf x}-{\bf
x^{n}})_{S^n}}^2+\norm{({\bf x}-{\bf
x^{n}})_{\overline{S^n}}}^2}$$ in terms of $\norm{{\bf x}-{\bf
x^{n-1}}}$. In contrast, in our treatment here, we try to upper
bound $\norm{({\bf x})_{\overline{U^{n}}}}$ in terms of
$\norm{({\bf x})_{\overline{U^{n-1}}}}$, for which we first obtain
an upper bound of $\norm{{\bf x}-{\bf x^{n}}}$ in terms of
$\norm{({\bf x})_{\overline{U^{n}}}}$ (i.e., \eqref{A1}) and
subsequently combine it with \eqref{A3} (with $n$ replaced by
$(n-1)$ in \eqref{A1}). Interestingly, in doing so, we can obtain
a decay relation of the metric $\norm{{\bf x}-{\bf x^{n}}}$ by
combining \eqref{A1} with \eqref{A3}, given as
\begin{equation}
\norm{{\bf x}-{\bf x^{n}}}\le\rho_{4K}\norm{{\bf x}-{\bf
x^{n-1}}}+(1-\rho_{4K})\alpha\norm{\textbf{e}},
\label{metric_song}
\end{equation}
where $(1-\rho_{4K})\alpha =
\frac{\sqrt{3(1+\delta_{3K})}}{1-\delta_{4K}} +
\sqrt{\frac{2(1+\delta_{4K})(1+2\delta_{4K}^2)}{1-\delta_{4K}^2}}$.
The relation \eqref{metric_song} is almost identical to the decay
relation given in Theorem 2 of \citep{Song}, with the only
difference being that the coefficient $(1-\rho_{4K})\alpha$ in
\eqref{metric_song} is lesser than the corresponding coefficient
in \citep{Song} as can be verified trivially.
For example, with $\rho_{4K}=0.5$ (meaning
$\delta_{4K} \approx 0.3$), $\alpha\approx 9.4$ in
\eqref{metric_song}, whereas $\alpha\approx 13.7$ in Theorem 2 of
\citep{Song}.

\end{document}